\documentclass[11pt]{article}

\usepackage{amsmath,amsfonts,amssymb,amsthm}
\usepackage{tikz}

\usetikzlibrary{arrows}
\usepackage{caption}

 \usepackage{graphicx}
 \usepackage{tikz}
\usetikzlibrary{arrows,positioning,fit,shapes,calc}
\usetikzlibrary{decorations.markings}

  \def\vertexscale{1}
  \def\arrowscale{1.2}
  \def\arrowstyle{latex'} 

  \pgfarrowsdeclarecombine{arrowstyle symmetric}{arrowstyle symmetric}{\arrowstyle}{\arrowstyle}{}{}

  \newcommand{\drawvertex}[1]{\path #1 coordinate (vtemp);\filldraw[scale=\vertexscale] (vtemp) circle (0.05)}

\newtheorem{theorem}{Theorem}[section]
\newtheorem{lemma}[theorem]{Lemma}
\newtheorem{proposition}[theorem]{Proposition}

\theoremstyle{definition}
\newtheorem{definition}[theorem]{Definition}
\newtheorem{example}[theorem]{Example}
\newtheorem{remark}[theorem]{Remark}

\title{Generalized parity proofs of the Kochen-Specker theorem}

\author{Petr Lison\v{e}k${}^1$, Robert Raussendorf${}^2$, 
Vijaykumar Singh${}^{1,2}$ \vspace{2mm} \\  
$^1$Department of Mathematics,  Simon Fraser University,  Burnaby, BC,\\ 
V5A 1S6, Canada\\
$^{2}$Department of Physics and Astronomy, University of British Columbia,\\
 Vancouver, BC, V6T 1Z1, Canada\\
\ \\
e-mail: {\tt plisonek@sfu.ca, raussen@phas.ubc.ca, vijay.k.1@gmail.com}
}

\begin{document}

\maketitle

\begin{center}
\end{center}

\begin{abstract}
We discuss two approaches to producing
generalized parity proofs
of the Kochen-Specker theorem.
Such proofs use contexts of observables
whose product is $I$ or $-I$; we call them constraints.
In the first approach, one starts with a fixed set of constraints
and methods of linear algebra are used to produce  subsets
that are generalized parity proofs.  
Coding theory methods are used for enumeration of the proofs by size.
In the second approach, one starts with the combinatorial structure
of the set of constraints and one looks for ways to suitably populate
this structure with observables. As well, we are able to show
that many combinatorial structures can not produce parity proofs.
\end{abstract}

\def\Z{{\mathbb Z}}
\def\C{{\mathbb C}}
\def\wt{{\rm wt}}
\def\la{\langle}
\def\ra{\rangle}

\def\B{\mathcal B}
\def\Co{\mathcal C}
\def\P{\mathcal P}
\def\I{\mathcal I}
\def\O{\mathcal O}

\def\z{{\bf 0}}
\def\o{{\bf 1}}

\section{Introduction}
\label{sec-intro}

In 1935, Einstein, Podolsky and Rosen asked the question of whether quantum mechanics can be considered a complete theory of physical phenomena \cite{EPR}. They ended by hinting at the possibility of classical descriptions of quantum mechanics in which the randomness of quantum measurement was modelled by a hidden probabilistic  parameter.  Hidden variable models (HVMs) turned out to be viable, as was demonstrated by Bohmian mechanics \cite{DBM}.  

However, if the seemingly most innocuous additional assumptions are made, HVMs can no longer reproduce the predictions of quantum mechanics. In a hidden variable model, in stark contrast to quantum mechanics, measurement outcomes exist prior to measurement, and are merely revealed. If the additional assumption is locality, then the Bell inequalities \cite{Bell} separate  quantum mechanics from all HVM descriptions. If the additional assumption is non-contextuality, the same is achieved by the Kochen-Specker Theorem~\cite{KS} (abbreviated KS Theorem henceforth).

Non-contextuality means the following: Let an observable $A$ be measured jointly with one of the compatible observables $B$ or $C$, and $B$  is incompatible with $C$. In quantum mechanics, two operators are said to compatible if and only if they commute.
 An HVM is non-contextual if the `pre-existing' measurement outcome $\mu(A)$ for  $A$ is independent of whether $A$ is measured jointly with $B$ or  with $C$. This seems a very reasonable requirement, since $A$ may be measured even before a decision has been made about whether to measure $A$ jointly with $B$ or with $C$. In other words, $A, ~B$ and $C$ can be all assigned measurement outcomes at once. Nevertheless, the assumption is of consequence.

\begin{theorem}[\cite{KS} Kochen and Specker]
\label{KST}
In Hilbert spaces of dimension\break $d \geq 3$, quantum mechanics cannot be described by any non-contextual hidden-variable model. 
\end{theorem}

For quantum mechanics, contextuality (i.e., the absence of  non-\break contextuality) is a feature that distinguishes it from classical physics. For quantum information theory, contextuality is also a resource. As an example, consider a quantum computer made of odd-dimensional  qudits. In such a setting, fault-tolerant quantum computation with distillation of so-called magic states \cite{BK} requires contextuality in order to be universal \cite{Eme}. 

Furthermore, proofs of the KS theorem can be translated into cryptographic protocols \cite{Cab} and into measurement-based quantum computations \cite{AB}. To make use of these correspondences, it is desirable to generate, classify and enumerate large numbers of KS proofs.
Millions of KS proofs have been identified, for example, in symmetric structures living in low-dimensional Hilbert spaces, such as the 600-cell \cite{WAMP}. 

In quantum mechanics, observables are represented by Hermitian
operators \cite{Sak}. Throughout this paper we consider {\em binary observables,} 
which are observables with eigenvalues 1 and $-1$.
A {\em context} is a set of pairwise commuting observables.
By a {\em constraint} we mean a set of pairwise commuting observables whose product 
is $I$ or $-I$.
In this paper, by a {\em parity proof of Kochen-Specker  theorem}
(abbreviated parity KS proof, or simply KS proof)
we mean a set $M$ of constraints such that 
each observable occurs in an even number of constraints in $M$,
and the number of constraints whose product is $-I$ is odd.
We explore the internal structure of such proofs
from two different viewpoints.
While most observable-based parity KS proofs occurring in literature
use Pauli observables (e.g., \cite{Mer,WA12,WA13}),
the methods discussed in this paper are {\em not} subject to this restriction.

In Section \ref{sec-MerminConstraint} we show that, given a set of constraints, the set of parity KS proofs constructed from these constraints
is in a one-to-one correspondence with a coset of a certain binary linear code. In particular, there are always 0 or $2^k$ such proofs
where $k \in \mathbb{N}$.
We discuss in more detail the case where the constraints
are derived from a set of intersecting orthogonal bases.
In this case our KS proofs  generalize the previous parity KS proofs 
(see, e.g., \cite{Per,WA11,WAMP}, and many other references)
since our method of deriving constraints from one orthogonal basis
is more general.
We observe (and give an illustration by an example)
that enumeration of parity proofs by size is possible indirectly
by an application of duality of vector spaces known from coding theory.
We present a method for finding parity proofs of small cardinality
(i.e., parity KS proofs consisting of a small number of constraints).

In Section \ref{Sec-Incidence-structures} we consider 
parity KS proofs that can be obtained from a given incidence structure.  In the incidence structure we let the points (vertices) correspond to observables 
and the blocks to constraints. 
The incidence structure then describes the intersection
pattern of putative constraints in a KS proof, and we ask
whether there exists an assignment of observables to points of
the structure that turns it into a KS proof.
We describe an algorithm that helps with answering this question
for a given incidence structure.
We consider the class of the simplest possible incidence structures
from which KS proofs can arise, namely the structures where
each point belongs to exactly two blocks
and each block contains exactly three points.
We decide the existence or non-existence of all such KS proofs
containing up to 15 observables (that is, up to 10 constraints).
The non-existence results apply to observables that
are of any dimension and are not necessarily Pauli observables.

\subsection{Notations}
\label{sec-obs}
We work in $\C^d$ with the inner product
$\la x,y\ra=\sum_{i=1}^d \overline{x_i}y_i$
where $\overline{z}$ denotes the complex conjugate of $z$.
Elements of $\C^d$ are considered as row vectors.
The transpose of a matrix $A$ is denoted $A^T$ and
the conjugate transpose of a matrix $A$ is denoted $A^{\dagger}$.
The symbols $\z$ and $\o$ denote the zero vector
and the all-one vector of an appropriate dimension.
Let $Y=\{y_1,\ldots,y_n\}$ be a set.
For any subset $Z\subseteq Y$
the {\em characteristic vector} of $Z$, denoted $\chi_Z$,
is defined by $(\chi_Z)_i=1$ if $y_i\in Z$
and $(\chi_Z)_i=0$ otherwise.
For a finite set $S$ let $|S|$ denote the number of elements of $S$.

A {\em graph} is a pair $G=(V,E)$ where $V$ is the set of {\em vertices}
and $E$ is the set of {\em edges,} which are unordered pairs
of vertices. We say that $u,v$ are the {\em endpoints} 
of the edge $\{u,v\}\in E$.
For $u,v\in V$, if $\{u,v\}\in E$,
then $u,v$ are called {\em adjacent.}
The {\em degree} of a vertex is the number of edges
to which it belongs (that is, the number of vertices
to which it is adjacent).  A {\em clique} in $G$ is a set $U\subseteq V$
such that any two distinct vertices in $U$ are adjacent.
Two graphs $G=(V,E)$ and $G'=(V',E')$ are {\em isomorphic}
if there exists a bijection $f:V\rightarrow V'$
such that $u,v$ are adjacent in $G$
if and only if $f(u),f(v)$ are adjacent in $G'$.
A graph $G=(V,E)$
is {\em connected} if for any $u,v\in V$, $u\neq v$,
there exists a sequence of vertices
$(w_0,w_1,\ldots,w_k)$ such that $w_i,w_{i+1}$
are adjacent for $i=0,\ldots,k-1$
and $w_0=u$, $w_k=v$.
For a positive integer $n$
the {\em complete graph} $K_n$
is defined as a graph on $n$ vertices such that any two distinct
vertices are adjacent.

Let $F$ be a field. For an $m\times n$ matrix $A$ over $F$,
the $F$-vector space of those $x\in F^n$ such that $Ax^T=\z$
is called the {\em kernel} of $A$ and denoted $\ker A$.

For $x\in\Z_2^n$ the number of non-zero coordinates of $x$
is called the {\em Hamming weight} of $x$, denoted $\wt(x)$.

\section{Parity proofs on a given set of constraints}
\label{sec-MerminConstraint}
Proofs of the Kochen-Specker (KS) theorem have been considerably simplified since they first appeared, and they come in various kinds. Some, such as the colouring proofs, are based on  interconnected orthogonal bases of $\C^d$ \cite{KS,Per}. The proofs that are of concern in this paper are based 
on interconnected contexts. Such proofs were first given by Mermin \cite{Mer}. Other proofs phrased in the framework of category theory \cite{Doer} also exist. 

To introduce the notion of a parity proof of the Kochen-Specker theorem, we  first review an example and then we give a general definition. 
Mermin's proof in $d=4$ \cite{Mer} invokes 9 Pauli observables in 6 contexts,
\begin{equation}
\label{PMsquare}
\begin{array}{ccc}
X\otimes I & I\otimes X & X\otimes X\\
I\otimes Y & Y\otimes I & Y\otimes Y\\
X\otimes Y & Y\otimes X & Z\otimes Z
\end{array},
\end{equation}
where 
\begin{equation}
\label{eq-XYZI}
I= \left( \begin{array}{cc} 1 & 0 \\ 0 & 1
\end{array} \right), \,X = \left( \begin{array}{cc} 0 & 1 \\ 1 & 0
\end{array} \right), \,
Z = \left( \begin{array}{cc} 1 & 0 \\ 0 & -1
\end{array} \right)\, , Y = \left( \begin{array}{cc} 0 & -i \\ i & 0
\end{array} \right).
\end{equation}
In (\ref{PMsquare}), the contexts are represented by the three rows and the three columns. Note that the observables in all rows and in all columns except the third column multiply to the identity $I$, whereas the observables in the third column multiply to $-I$. 

 In the above example, it is impossible to assign  `pre-existing' values to the observables, which can be seen as follows. Assume an assignment exists. First, the `pre-assigned' measurement outcomes $\mu(\cdot)$ must all be $+1$ or $-1$ which are the eigenvalues of the observables in question. 
In the quantum mechanics, if observables $A_1, A_2,\ldots,A_n,A_{n+1} $ belonging to a context satisfy $A_1A_2\cdots A_n=A_{n+1}$,  then their measurement outcomes satisfy $\mu (A_1) \mu (A_2)\cdots \mu(A_n) =\mu(A_{n+1})$,
see \cite[Section II]{Mer}.

Hence,  the six constraints among the observables translate into corresponding constraints on the pre-assigned values $\mu$. For example, 
$(X\otimes I)(I\otimes X)(X\otimes X) = I$ implies 
$\mu(X\otimes I)\mu(I\otimes X)\mu(X\otimes X) = 1$. Let us now work out the product of the values $\mu$ over all observables in any given context, and then over all contexts. The constraining relations give us the products within the six contexts, five times 1 and once $-1$. The product over all contexts is thus $-1$. However, we may work out this product differently. We observe that every value $\mu$ appears in exactly two contexts. Since $\mu=\pm1$ for all observables, the product of all the $\mu$
over all contexts must therefore be $+1$. Contradiction! No assignment of values $\mu$ to the observables in (\ref{PMsquare}) exists.


\subsection{Structure of parity proofs}
\label{subsec-gen-parity}

Recall from Section~1
that a {\em context} is a set of pairwise commuting observables.
For a context $C$ denote $P(C)=\prod_{O\in C} O$
the product of all observables in $C$.
Recall that a {\em constraint} was defined to be any context $C$
such that  $P(C)=I$ or $P(C)=-I$.

The following definition generalizes Mermin's approach \cite{Mer}
to proving the KS theorem that was outlined in the example given above.
This approach was subsequently applied also in other papers (e.g., \cite{WA12}).

\begin{definition}
\label{def-MKS}
A {\em parity proof of the Kochen-Specker theorem}
(parity KS proof) is a set $M$ of constraints such that
each observable occurs in an even number of constraints,
and the number of constraints $C\in M$ such that $P(C)=-I$ is odd.
The {\em size} of this parity KS proof is
defined to be the cardinality of the set $M$ (the number of constraints).
\end{definition}

Some authors use the term ``parity proof'' in a more narrow sense,
and we will return to this issue in Section \ref{sec-ray-proofs} below.

To see why any parity KS proof as introduced in Definition \ref{def-MKS} 
actually
proves the KS theorem, one generalizes the idea
of Mermin's original argument that we reviewed above.
We know that 
for each constraint $C$ the measurement outcomes $\mu(O)$
satisfy $\prod_{O\in C} \mu(O)=\pm 1$
where $\prod_{O\in C} O=\pm I$
and the same sign occurs in both equalities.
Let us compute the product $Q=\prod_{C}\prod_{O\in C} \mu(O)$
over all constraints $C$ occurring in the parity KS proof
in two different ways,
as was done in Mermin's example above.
One of these calculations shows that $Q=-1$, whereas
the other one yields $Q=1$, thus proving the Kochen-Specker theorem.

It has been noticed repeatedly
that the number of certain parity proofs associated with a given
set of orthogonal bases
is 0 or a power of 2,
but apparently this has never been explained.
(See, for example, Section 4.1 in \cite{WA12}.)
In the next theorem we advance this observation to a more general
setting and we prove it.

\begin{theorem}
\label{thm-Merm}
Let $\Co$ be an arbitrary finite set of constraints.
Then the number of parity proofs of the Kochen-Specker theorem
that are subsets of $\Co$
is 0 or $2^k$ for some non-negative integer $k$.
\end{theorem}
\begin{proof}
Let $\Co=\{C_1,\ldots,C_n\}$.
Let $\O=\{O_1,\ldots,O_m\}$ be the set
of those observables that occur in at least two constraints in $\Co$.
Let $H$ be the $m\times n$
matrix over $\Z_2$ defined as follows.
We set $H_{i,j}=1$ if $O_i\in C_j$
and $H_{i,j}=0$ if $O_i\not\in C_j$.
Further we define the vector $p\in\Z_2^n$
by letting $p_j=0$ if $P(C_j)=I$ and $p_j=1$ if $P(C_j)=-I$.
Let $H'$ be the $(m+1)\times n$ matrix obtained
by appending $p$ to $H$ as the last row.
With any subset $\Co'\subseteq \Co$ we associate 
its characteristic vector $\chi_{\Co'}\in\Z_2^n$.
Then $\Co'$ is 
a parity proof of the Kochen-Specker theorem 
if and only if $H'\chi_{\Co'}^T=(0,\ldots,0,1)^T$.
The set of such vectors $\chi_{\Co'}$ either is empty or it is a coset of
$\ker H'$ in $\Z_2^n$.  In the latter case, the cardinality
of this coset equals the cardinality of $\ker H'$,
which is $2^k$ where $k$ is the dimension of $\ker H'$.
\end{proof}

\begin{remark}
\label{rem-list-samp}
Once a basis for $\ker H'$ has been determined,
an efficient exhaustive listing of parity proofs is possible,
for example in a Gray code ordering \cite{Sav}. 
Also, uniform sampling of the proofs becomes easy.
\end{remark}

\subsection{Parity proofs associated with a set of orthogonal bases}
\label{subsec-ONB}

For an application of Theorem \ref{thm-Merm}
we need a set of constraints.
One possible way of constructing a set of constraints
is as follows.  One starts with a set $\B$ of orthogonal bases
of $\C^d$.  Let us assume for a moment that $\B$ is a set of 
{\em orthonormal} bases, although we will see shortly
that this restriction can be removed easily.
For each basis $B\in\B$,
say $B=\{b_1,\ldots,b_d\}$,  and for each $\lambda \in\Z_2^d$
one constructs the observable
\begin{equation}
\label{eq-OBl-1}
O_B(\lambda)=\sum_{i=1}^d (-1)^{\lambda_i} b_i^{\dagger} b_i.
\end{equation}
In this subsection we analyze how such observables
can be combined to produce constraints and subsequently 
we analyze how parity KS proofs can be constructed from such constraints.

Note that it is sufficient to start with a set of orthogonal bases
and compute the observables $O_B(\lambda)$ by
\begin{equation}
\label{eq-OBl-2}
O_B(\lambda)=\sum_{i=1}^d (-1)^{\lambda_i} \frac{b_i^{\dagger} b_i}{\la b_i,b_i\ra}.
\end{equation}
When doing exact computations in computer algebra systems such as Maple or Magma, 
the formula (\ref{eq-OBl-2}) is superior
to (\ref{eq-OBl-1}) as it avoids the unnecessary introduction
of the square roots needed to normalize the $b_i$.
In practice, one will precompute
the matrices $P_i:=\frac{b_i^{\dagger} b_i}{\la b_i,b_i\ra}$
and then $O_B(\lambda)$ are found as signed sums of the $P_i$.

Note that $O_B(\lambda+\mu)=O_B(\lambda)O_B(\mu)$ for all $B,\lambda,\mu$,
 hence $O_B(\lambda)$ and $O_B(\mu)$ commute for all $B,\lambda,\mu$.
Also, all eigenvalues of $O_B(\lambda)$ are $1$ or $-1$,
in accordance with our definition of observable.
Further, $O_B(\lambda+\o)=-O_B(\lambda)$.
Consider a subset $T\subseteq \Z_2^d$.
Then 
$\prod_{\lambda\in T} O_B(\lambda)=O_B(\sum_{\lambda\in T} \lambda)$,
hence 
$\prod_{\lambda\in T} O_B(\lambda)=I$ if and only if
$\sum_{\lambda\in T} \lambda =\z$
and
$\prod_{\lambda\in T} O_B(\lambda)=-I$ if and only if
$\sum_{\lambda\in T} \lambda  =\o$.
Thus, for a fixed $B$, a set $\{ O_B(\lambda) : \lambda\in T \}$
is a constraint if and only if 
$\sum_{\lambda\in T} \lambda$ is $\z$ or $\o$.
For the purpose of constructing constraints, 
one can restrict attention
to the set of vectors
\[
K := \{ \lambda \in\Z_2^d : \lambda_1=0,\:  \lambda\neq \z \}
\]
and further restrict to subsets $T\subseteq K$
such that $\sum_{\lambda\in T} \lambda =\z$.
This can be justified as follows:
Any constraint 
which is of the form $\{O_B(\lambda) : \lambda \in T\}$
with unrestricted vectors $\lambda$ 
can be transformed to a constraint of the same form
with vectors $\lambda$ restricted to the set $K$
by adding $\o$ to some of the $\lambda$s,
which only flips the sign of the corresponding observables
$O_B(\lambda)$, hence this transformation
applied to a constraint produces again a constraint.
By Definition \ref{def-MKS},
in any parity KS proof each observable occurs an even number of times.
Thus, there is an odd number of constraints with product $-I$
{\em before} the sign flip 
if and only if
there is an odd number of constraints with product $-I$
{\em after} the sign flip.

By Definition \ref{def-MKS}, any observable occurring in a parity KS proof
must occur in at least two constraints.  Thus for each $B\in\B$
we can further
restrict attention to the observables $O_B(\lambda)$
whose vectors $\lambda$ belong to the set
\[
L_B := \{ \lambda\in K : 
(\exists B'\in\B,\, B'\neq B)(\exists \lambda'\in K)
O_B(\lambda)=\pm O_{B'}(\lambda') \}.
\]
Let $n_B=|L_B|$ for each $B\in\B$
and without loss of generality assume $n_B>0$ for all
$B\in\B$.  This means that at this point
we delete from the computation
all $O_B(\lambda)$ for which $\lambda\not\in L_B$
and also we delete from $\B$ all $B\in\B$ 
for which $n_B=0$ (and their $O_B(\lambda)$s).

The constraints associated with any $B\in \B$ are
precisely of the form $\{ O_B(\lambda): \lambda \in T\}$
 where
$\emptyset\neq T\subseteq L_B$ and  $\sum_{\lambda \in T} \lambda=\z$.
Thus, define
\[
U_B := \left\{ T \ :\  T\subseteq L_B,\: \sum_{\lambda \in T} \lambda=\z \right\}.
\]
For each $B$, the set $U_B$ is a vector space  over $\Z_2$.
This can be seen by representing an element $T\in U_B$ 
by its characteristic vector $\chi_T$,
once we fix an ordering (labeling) of the set 
$L_B=\{\lambda^1,\lambda^2,\ldots,\lambda^{n_B}\}$.
In this representation, $U_B$ is a subspace of $\Z_2^{n_B}$.

In the computer implementation of the algorithm one now identifies
the $O_B(\lambda)$ that are equal up to the sign, across all bases  $B\in\B$.
Any such class of observables is henceforth treated as one single observable.

One can now construct parity KS proofs by taking 
\begin{equation}
\label{eq-Co}
\Co := \bigcup_{B\in\B} \{ \{ O_B(\lambda) : \lambda \in T \} : T \in U_B, T\neq\emptyset \} 
\end{equation}
in Theorem \ref{thm-Merm}.
The proof of Theorem \ref{thm-Merm}
and Remark \ref{rem-list-samp} then allow us to exhaustively list
and/or uniformly sample from the set of parity KS proofs associated with $\B$.

Due to  the fact that experimental realizations
of KS proofs are presently limited to KS proofs of small size,
it is interesting to specifically address finding parity KS proofs
with few constraints.  Upon rereading the proof of Theorem \ref{thm-Merm},
this is equivalent to 
finding vectors $x$ 
of small Hamming weight satisfying $H'x^T=(0,0,\ldots,0,1)^T$.
Such vectors can be found using the {\em meet-in-the-middle}
idea:  If one looks for vectors of Hamming weight at most $w$,
then it is sufficient to compute,
for all vectors $y$ of Hamming weight at most $\lceil w/2 \rceil$,
the pairs $(y,H'y^T)$  and store them in a table.
This table is indexed by the second components of the pairs.
Upon setting $x=y^1+y^2$ where $y^1$ and $y^2$ have weight 
at most $\lceil w/2 \rceil$, 
the condition $H'x^T=(0,0,\ldots,0,1)^T$
becomes equivalent to $H'(y^1)^T+H'(y^2)^T=(0,0,\ldots,0,1)^T$.
For each vector $y$ of Hamming weight at most $\lceil w/2 \rceil$
one stores the pair $(y,H'y^T)$ in the table
and at the same time one queries the table for the
existence of a pair (or pairs) 
of the form $(t,u)$ where $u=H'y^T+(0,0,\ldots,0,1)^T$.
If such pair(s) is/are found in the table, then $y+t$
defines a parity KS proof of cardinality at most $w$.

If one is not interested in finding just small parity KS proofs,
but rather one seeks to have access to all parity KS proofs
derived from a set $\B$ of orthogonal bases,
then the method outlined in this section
should be modified in the following way.  
Let $B\in\B$ be fixed.
Instead of creating
one constraint (and thus one column of matrix $H'$) for each 
non-empty element
of $U_B$, the columns of $H'$ are created to correspond to a basis
of $U_B$.  This often makes the matrix $H'$ much smaller
(however the Hamming weight of a solution
to $H'x^T=(0,0,\ldots,0,1)^T$ no longer corresponds
to the size of the KS proof that it represents).

\subsection{Parity proofs based on rays}
\label{sec-ray-proofs}

Let $\B$ be a set of orthogonal bases in $\C^d$.
For any $B\in\B$ and $v\in B$ let $S_v:=I-2\frac{v^{\dagger}v}{\la v,v\ra}$.
Note that $S_v=O_B(\lambda)$ for a certain vector $\lambda$
of weight 1.  In this way we associate to each vector $\lambda$
of weight 1 the 1-dimensional subspace of $\C^d$ spanned by $v$
and called a {\em ray.}
      
In some papers the concept of ``parity proof''
is used for the special type of KS proof in which 
rays are the primary objects.  To connect with our Definition \ref{def-MKS},
introduce the constraint
\begin{equation}
\label{ref-eq-CB}
C_B := \{ O_B(\lambda) : \lambda\in\Z_2^d,\; \wt(\lambda)=1 \}
\end{equation}
for each $B\in\B$.
Note that $\prod_{O\in C_B} O=-I$ for each $B$. 
After making this connection,
we see that a special type of a parity KS proof,
which we will call a {\em ray parity proof} in this paper,
is obtained
from any set $\B$ of orthogonal bases of $\C^d$ 
that satisfies the following two conditions:
The cardinality of $\B$ is odd and each ray in $\C^d$
belongs to an even number of bases in $\B$.

As an immediate corollary of Theorem \ref{thm-Merm}
we see that the number of ray parity proofs arising from
$\B$ is 0 or $2^k$ for some non-negative integer $k$.

A common way of constructing ray parity proofs
is as follows.  One starts with a set of rays, say $R$.
More precisely, $R$ is a set of vectors spanning the rays.
One finds
the set of all orthogonal bases that are subsets of $R$, let us call
this set $\B_R$. 
Computationally, this can be done as follows:
Form the {\em orthogonality graph} $G_R$ whose vertices
are elements of $R$ and two vertices are adjacent if and only if
the corresponding rays are orthogonal.  An orthogonal basis of $\C^d$
that is a subset of $R$ corresponds to a clique of size $d$ in $G_R$.
This reduces the problem 
of constructing the set $\B_R$
to finding all cliques of size $d$ in $G_R$,
which can be handled for example by the very efficient clique finder
available in Magma \cite{Mag}.
Once the set $\B_R$ is obtained,
ray parity proofs can be constructed using Theorem \ref{thm-Merm},
in which for the observables $O_i$ we take all possible $S_v$
($v\in R$) and for contexts $C_j$ we take all $C_B$ ($B\in \B_R$)
as defined in (\ref{ref-eq-CB}).

\begin{remark}
\label{remark-ray}
Since the product of each constraint $C_B$ defined in (\ref{ref-eq-CB})
equals $-I$, for finding ray parity proofs
one can simplify the proof of Theorem \ref{thm-Merm}
to considering the matrix $H$ only, and then considering 
odd weight vectors in $\ker H$. 
\end{remark}

\subsection{Duality and weight distributions}
\label{sec-dua}

Let $x\cdot y=xy^T$ denote the usual inner
product on $\Z_2^n$ and for a subspace $S$ of $\Z_2^n$
let $S^{\perp}$ denote the {\em dual} of $S$ defined as
\[ S^{\perp}:=\{ x\in\Z_2^n : (\forall y\in S) x\cdot y=0\}. \]
Then $S^{\perp}$ is a subspace of $\Z_2^n$ of dimension $n-\dim S$
and $S^{\perp\perp}=S$.

For a subspace $S$ of $\Z_2^n$ let 
$A_i$ denote the number
of vectors of weight $i$ contained in $S$
and let
$B_i$ denote the number
of vectors of weight $i$ contained in $S^{\perp}$.
The sequences $(A_i)$ and $(B_i)$
are called {\em weight distributions} of $S$ and $S^{\perp}$ respectively.
The famous {\em MacWilliams Theorem} of coding theory
(see, for example, \cite[Chapter 5, Theorem 1]{MS})
gives a compact and easy to evaluate formula for computing the sequence
$(B_i)$ if the sequence $(A_i)$ is known (or vice versa, of course).

Moreover, 
a generalized version of the MacWilliams Theorem \cite{AM} 
allows one to compute the weight distribution
of a coset $a+D$ of $D$
in terms of the weight distributions of the codes $D^{\perp}$ 
and $D^{\perp}\cap \langle a\rangle^{\perp}$
where $a$ is a vector and $\langle a\rangle$ the subspace spanned by it.
We use this result with letting $D:=\ker H'$.
Thus, it is possible to count parity KS proofs by their size
{\em without} constructing them, and to do so in a way
that may be much faster than
exhaustively listing $\ker H'$ or  running the meet-in-the-middle
computation described in Section \ref{subsec-ONB}.
This approach is attractive in cases when the dimension of 
$(\ker H')^{\perp}$ is smaller than the dimension of $\ker H'$.

\subsection{Examples}
\label{sec-exa}

We illustrate the results 
of Sections \ref{subsec-gen-parity} through \ref{sec-dua} on two examples.
We use the computer algebra system Magma \cite{Mag}
for all computations.
Timings given below were obtained using Magma~2.19
running on Intel Core i7 CPU at 2.67~GHz.

\begin{example}
{\em Parity proofs in the 600-cell}

With notation as 
in Section \ref{sec-ray-proofs} let $R$ be the set of $60$ rays in $\C^4$
defined
by the vertices of the 600-cell \cite{WAMP}. As is known
there are $75$ orthogonal bases that are subsets of $R$ \cite{WAMP}.
We use Magma to find the 75 orthogonal bases.
By Theorem \ref{thm-Merm} and Remark \ref{remark-ray} we find
that there are precisely $2^{33}$ ray parity proofs of the Kochen-Specker theorem
found in the 600-cell. 


By another computation we found that each parity proof
arising from the 600-cell 
by the methods of Section \ref{subsec-ONB}
is a ray parity proof.  That is, using 
$O_B(\lambda)$ with $\wt(\lambda)>1$ does not yield any additional
parity proofs beyond what is listed above.
These two computations 
take only a few seconds.
\end{example}

\begin{example}
\label{ex-60-105}
{\em Parity proofs in the 60--105 system}

The $60$ rays in this system are introduced in \cite{WA11}
as joint eigenvectors of sets of commuting Pauli observables on two qubits.
These rays form $105$ orthogonal bases \cite{WA11}.
For simplicity denote $V:=\ker H$.
As in the previous example, it takes only a fraction of a second
to find the 105 bases and to find that the dimension of $V$ is 65
in this example. As we see vectors of odd weight in the basis for $V$,
we conclude that
there are precisely $2^{64}$ ray parity proofs of the Kochen-Specker theorem
found in the 60--105 system.  

This example allows us to illustrate an application 
of the material in Section \ref{sec-dua}.
While {\em listing} all $2^{64}$ ray parity proofs
would perhaps take hundreds of years of CPU time, we can still
compute their distribution according
to the number of bases (i.e., according to the number of constraints) 
that they involve in just a few hours, as follows. 

In our example $\dim V^{\perp}=105-65=40$ and listing all vectors
in $V^{\perp}$ is feasible (it takes a few hours);
this allows us to compute the weight distribution of $V^{\perp}$.
Then we can compute
the weight distribution of $V=(V^{\perp})^{\perp}$ using MacWilliams Theorem.
This method in fact is built into the {\tt WeightDistribution} function
in Magma.  We show the output in Appendix A below.
We conclude that the 60--105 system contains 160 ray parity proofs
involving 9 bases (constraints), 
18240 ray parity proofs involving 11 bases (constraints), and so on.
The computation takes less than 8 hours of CPU time in Magma.

By another calculation we determine that there are $2^{439}$
parity KS proofs obtained by the methods of Section \ref{subsec-ONB}
from the 60--105 system. 
In order to 
find the distribution of these proofs by size (number of constraints),
one could apply the generalized MacWilliams Theorem
(second part of Section \ref{sec-dua}).  This would amount
to computing the weight distribution of a subspace of $\Z_2^{495}$
of dimension $495-439=56$. This computation is beyond our resources.
However, a similar  calculation would be feasible 
for somewhat smaller cases, for example by taking a suitable
subset of the 60--105 system.
\end{example}

\begin{example}
{\em Ray parity proofs in the root system $E_8$}

In this example we take for $R$ the $120$ rays in $\C^8$
determined by the 120 pairs of roots of the lattice $E_8$.
In 0.3 second we find that there are $2025$ orthogonal bases,
and in another 0.3 second we do the linear algebra step
to find that the dimension of $V$ is 1941.
As there are odd weight vectors in the basis for $V$,
we conclude that 
there are $2^{1940}$ ray parity proofs in the $E_8$ root system.
\end{example}

\section{Parity proofs on an incidence structure}
\label{Sec-Incidence-structures}

Recall that we work with observables
that have eigenvalues $1$ and $-1$.
Thus we have $A^2=I$ for each such observable $A$.
While this class of observables contains all tensor products of Pauli operators,
we want to emphasize that the results of this section
are not limited to Pauli operators.

Now we give a completely different strategy for producing
 parity proofs. In Section \ref{subsec-gen-parity} we started
with a given set of constraints (consisting of fixed observables)
and we were finding its subsets
which are  parity proofs.  We now revert this process.
We start with a combinatorial structure of the set of constraints
and we are asking if and how this structure subsequently
can be populated with observables to  produce
 parity proofs.

\subsection{Incidence structures}

\begin{definition}
\label{def-inc-str}
Let $\P=\{v_1,\ldots,v_n\}$ be a finite set whose elements
we call {\em points}. An {\em incidence structure}
is a pair $(\P,\B)$ where $\B$ is a set of subsets of $\P$ (called {\em blocks})
such that each point in $\P$ occurs in an even number of blocks in $\B$.
Also each block contains at least three points.  
\end{definition}

Here points model observables (note that the same observable
may be assigned to distinct points) and blocks model constraints.
Both requirements on the incidence structure are directly linked
to the definition of a  parity proof.
The requirement on block size
follows from the fact that a constraint must 
contain at least three observables.

One class of incidence structures can be produced from cubic graphs.
A {\em cubic graph} is a graph in which each vertex has degree 3.
Clearly the number of vertices in a cubic graph must be even.
Exhaustive lists of connected cubic graphs on up to 14 vertices
(up to isomorphism) along with their drawings can be found
in \cite[pp.~126--144]{RW}.
Magma incorporates B.D. McKay's system {\em nauty} \cite{nauty}  
that contains a very efficient procedure
for generating graphs up to isomorphism, and we have used it in our computations.
The generation can be restricted by vertex degree, 
number of edges, connectedness, and
many other graph parameters.

Given a connected cubic graph $G=(V,E)$, we produce the  incidence structure
$I=(E,\B)$ such that the points of $I$ are the edges of $G$
and for each vertex $v\in V$ of $G$
there is exactly one corresponding block $b_v\in\B$ in $I$.
The block
$b_v$ contains precisely the points of $I$
that represent those edges of $G$ whose one endpoint is $v$.
Since $G$ is connected, $G$ can not be decomposed into the union
of smaller graphs on disjoint vertex sets.
It follows that the incidence structure $I$ constructed from $G$
does not decompose into a disjoint union of smaller incidence structures.

Note that incidence structures constructed in this way from connected cubic graphs
are the {\em smallest (or simplest)} incidence structures
relevant to our paper in the sense
that each point belongs to exactly two blocks and each block has
size exactly three. However, many other incidence structures
can be constructed as well, for example by starting from
graphs in which each vertex has degree {\em at least} three,
or more generally starting from hypergraphs.


\begin{example}
\label{ex-Pasch-conf}
Consider the complete graph $K_4$ on the vertex set $\{1,2,3,4\}$.
Label its edges as 
$e_1=\{1,2\}$, $e_2=\{1,3\}$, $e_3=\{1,4\}$, 
$e_4=\{2,3\}$, $e_5=\{2,4\}$, $e_6=\{3,4\}$. 
The procedure given above produces the incidence structure $(\P,\B)$
whose set of points
is $\P=\{e_1,e_2,e_3,e_4,e_5,e_6\}$ and the set of blocks is
$\B=\{b_1,b_2,b_3,b_4\}$ where
$b_1=\{e_1,e_2,e_3\}$, $b_2=\{e_1,e_4,e_5\}$, 
$b_3=\{e_2,e_4,e_6\}$, $b_4=\{e_3,e_5,e_6\}$.
This  incidence structure is well known as the {\em Pasch configuration.}
\end{example}

Given an incidence structure, we ask if there is an assignment
of observables to its points such that blocks become constraints
and the incidence structure becomes a  parity proof.
This question can be answered using the following lemma.

\begin{lemma}
\label{lem-IncS-Merm}
Let $\I=(\P,\B)$ be an incidence structure. Let us assign to each point $p\in\P$
the observable $O(p)$ and suppose that under this assignment
each block in $\B$ becomes a constraint.
Then under this assignment $\B$ becomes a  parity proof if and only if
\begin{equation}
\label{eq-pro-pro}
\prod_{b\in\B} \prod_{p\in b} O(p) = -I.
\end{equation}
\end{lemma}
\begin{proof}
According to Definition \ref{def-MKS} we only need to check
that the number of constraints whose product is $-I$ is odd.
This happens if and only if the product of products of all constraints is $-I$.
\end{proof}

\begin{example}
\label{ex-Pasch-not-Mermin}
We now show that
the Pasch configuration as introduced in Example \ref{ex-Pasch-conf}
can not produce  parity proofs.  
Suppose that we assigned observables to points of the Pasch configuration.
Recall that $A^{-1}=A$ for each observable $A$.
Without loss of generality, after the assignment of observables 
the Pasch configuration has the following form:
\begin{eqnarray*}
&&\{O_1,O_2,s_1O_1O_2\}, \{O_1,O_3,s_2O_1O_3\}, 
\{O_2,O_3,s_3O_2O_3\},\\ 
&&\{s_1O_1O_2,s_2O_1O_3,s_3O_2O_3\}
\end{eqnarray*}
where
$O_1,O_2,O_3$ are observables
and $s_1,s_2,s_3\in\{-1,1\}$, see Figure \ref{fig: pasch}.

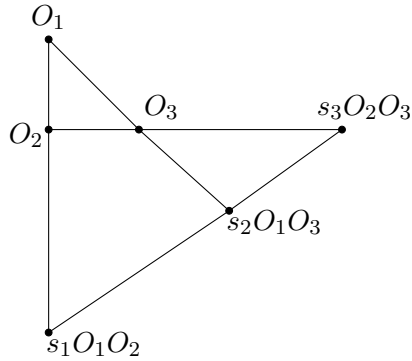
\begin{figure}[htb]
\begin{center}
  \def\vertexscale{3}
  \def\arrowscale{1.2}
  \def\arrowstyle{latex'}
\begin{tikzpicture}[scale=0.3]

  \path (1,18) coordinate  (n6) ;
   \path (1,14) coordinate (n5)  ;
   \path (1,5) coordinate  (n4)  ;
   \path (5,14) coordinate   (n3) ;
   \path (9,10.4) coordinate (n2) ;
   \path  (14,14) coordinate (n1) ;

\begin{scope}
  \foreach \from/\to in {n6/n5,n5/n4, n5/n3, n3/n1, n4/n2, n2/n1, n6/n3, n3/n2}
    \draw (\from) -- (\to);

\drawvertex{(n6)};
\drawvertex{(n5)};
\drawvertex{(n4)};
\drawvertex{(n3)};
\drawvertex{(n2)};
\drawvertex{(n1)};

    \draw (n6)++(0,1) node {$O_1$};
    \draw (n5)++(-1,-0.25) node {$O_2$};
    \draw (n4)++(2,-0.5) node {$s_1O_1O_2$};
    \draw (n3)++(1,1) node {$O_3$};
    \draw (n2)++(2,-0.5) node {$s_2 O_1 O_3$};
    \draw (n1)++(1,1) node {$s_3 O_2 O_3$};
\end{scope}
\end{tikzpicture}
\caption{Pasch configuration}
\label{fig: pasch}
\end{center}
\end{figure}

The left-hand side of the product (\ref{eq-pro-pro}) equals
\begin{eqnarray*}
(O_1O_2s_1O_1O_2) (O_1O_3s_2O_1O_3)
(O_2O_3s_3O_2O_3) (s_1O_1O_2s_2O_1O_3s_3O_2O_3)\\
=s_1^2s_2^2s_3^2 O_1^6O_2^6O_3^6=I
\end{eqnarray*}
where we used that $s_i^2=1$ for all $i$
and $O_i^2=I$ for all $i$. Also we used that any two $O_i$
commute, because any two $O_i$ occur together in some block of
the configuration.  

Note that Lemma \ref{lem-IncS-Merm} has the following consequence:
If the left-hand side  of the product (\ref{eq-pro-pro})
equals $I$ under {\em every} permissible assignment of observables
to points of $\I$, then no parity proof can be constructed from $\I$.
Thus in particular we conclude
that no assignment of observables to points of the Pasch configuration
produces a  parity proof.
\end{example}

\subsection{Finitely presented groups}
\label{subs-fpg}

The hand calculations done in Example \ref{ex-Pasch-not-Mermin}
quickly become complicated when
the size of the incidence structure increases.  Therefore one desires
an {\em automated} 
process, which can be implemented on a computer,
for finding suitable assignments of operators to points
of the incidence structure such  that a  parity proof is produced,
or proving that no such assignment exists.  This indeed is possible
as we show next.

A {\em free group} (written multiplicatively)
on generators $g_1,g_2,\ldots$ 
is the group that has as its elements all possible (associative)
products of the $g_i$s and their inverses and all such products
are assumed to be distinct. A {\em finitely presented group}
is a free group modulo a set of relations, each of which
has the form of an equality of two elements of the free group
(i.e., an equality of two products of powers of $g_i$). 
The {\em word problem} in a finitely presented group $G$
is the question whether two elements of $G$, both written
as products of powers of $g_i$, are equal modulo the set of relations
used to define $G$.

Fix an incidence structure $\I$ as in Definition \ref{def-inc-str}
and introduce
an assignment of operators to points of $\I$ such that each
block becomes a constraint.
This is done as in Example \ref{ex-Pasch-not-Mermin} above,
and it involves introducing scalars $s_i$ and observables $O_i$.
Let $A_p$ denote the observable assigned to point $p$;
each $A_p$ is a (generally non-commuting) product of some
of the $s_i$ and some of the $O_i$.
This assignment can be done as follows:  In each step,
take a point $p$ that does not have $A_p$ assigned yet.
If there exists a block $b$ containing $p$ such that all points of $b$
except $p$ have their observables assigned already, then the assignment of $A_p$
is forced up to $\pm 1$, which is accounted for by an introduction
of the scalar $s_p$.  Otherwise, let $A_p:=O_i$ where $O_i$ is a new generator.

Ultimately, we are interested in the the left-hand side  of 
product (\ref{eq-pro-pro})
and since each $s_i$ equals $\pm 1$
and it appears with an even exponent in the left-hand side  of (\ref{eq-pro-pro}),
it will have no contribution to the left-hand side  of (\ref{eq-pro-pro}).
Hence we discard the $s_i$ from our computations and we only
compute with the $O_i$.

We now think of the finitely presented group $G_{\I}$
on the generators $O_i$ modulo the set of relations
that are of the following two types.
(Note that the identity matrix $I$
is the identity element of $G_{\I}$.)

(i) For each block $B$ of $\I$ we introduce the relations
$A_uA_v=A_vA_u$ for any $u,v\in B$, $u\neq v$.

(ii) For each point $p$ of $\I$ we introduce the relation $A_p^2=I$,
since throughout the paper 
we deal only with observables that have eigenvalues $\pm 1$.

Note that the construction of  $G_{\I}$ depends
on how the assignments of $A_p$ were done,
and there is a considerable freedom in choosing those assignments.

\subsection{Knuth-Bendix algorithm}

The left-hand side  of (\ref{eq-pro-pro}),
being a product of constraints, always equals $I$ or $-I$.
By Lemma \ref{lem-IncS-Merm}
and the discussion in Section \ref{subs-fpg},
if the left-hand side  of (\ref{eq-pro-pro})
can be shown to be equal to $I$ in the group $G_{\I}$, then
$\I$ can not produce a  parity proof.

Let $F$ be a free group and $G$ a finitely presented group
created from $F$ as described in Section \ref{subs-fpg}.
{\em Knuth-Bendix algorithm} can be used to solve the word problem
in $G$.  An introductory exposition on Knuth-Bendix algorithm
can be found, for example, in \cite[Chapter 12]{HEE}.
We used the implementation of Knuth-Bendix algorithm  found
in Magma \cite{Mag}.  
Knuth-Bendix algorithm solves the word problem 
in $G$ by creating a rewriting system for $G$.
Any element $a\in F$ is reduced to its unique {\em canonical form,}
let us denote it $c(a)$,
by a repeated application of this rewriting system.  The crucial property
of this rewriting system is that for any $a,b\in F$
we have $a=b$ in $G$ if and only if $c(a)=c(b)$.

There are two ways in which we use Knuth-Bendix algorithm:

Firstly, we ask whether the left-hand side  of (\ref{eq-pro-pro})
is equal to $I$ in  $G_{\I}$.  If that is the case, we stop
with the conclusion that no  parity proofs
can be produced from $\I$. 

Otherwise, we augment the relations
that were used to define $G_{\I}$, as listed under (i) and (ii)
near the end of Section \ref{subs-fpg}, 
by new relations of the form $AB=BA$ where $A,B$
are some two observables assigned to points of $\I$. 
Thus we obtain a new group $G'_{\I}$.  If the left-hand side of (\ref{eq-pro-pro})
was not equal to $I$ in $G_{\I}$,
but it becomes equal to $I$ in $G_{\I}'$, then we know that
$A$ and $B$ must {\em not} commute in any assignment of observables
to points of $\I$ that produces a  parity proof on $\I$.
Such information can be then used to prune the search
for assignments of observables to points of $\I$.

\begin{example}
\label{ex-six-vertices}
Up to isomorphism there are exactly two connected cubic graphs
on six vertices, see \cite[page 127]{RW} where they are depicted
as graphs C2, C3. 
The graph C3 leads, by the general construction that we gave above, to the 
incidence structure that is depicted as a $3\times 3$
grid in Figure \ref{fig:k5problems}.
\begin{figure}[htb]
\begin{center}
  \def\vertexscale{3}
  \def\arrowscale{1.2}
  \def\arrowstyle{latex'}
\begin{tikzpicture}[scale=0.4]

  \path (1,18) coordinate (n1);
  \path (5.5,18) coordinate (n2) ;
  \path (10,18) coordinate (n3);
  
  \path (1,14) coordinate (n4);
  \path (5.5,14) coordinate (n5) ;
  \path (10,14)coordinate (n6);
  
  \path (1,10) coordinate (n7) ;
  \path (5.5,10) coordinate (n8) ;
  \path (10,10) coordinate (n9) ;
\begin{scope}
  \foreach \from/\to in {n1/n2,n2/n3, n4/n5, n5/n6, n7/n8, n8/n9, n1/n4,n4/n7, n2/n5, n3/n6, n5/n8, n6/n9}
  \draw (\from) -- (\to);

\drawvertex{(n1)};
\drawvertex{(n2)};
\drawvertex{(n3)};
\drawvertex{(n4)};
\drawvertex{(n5)};
\drawvertex{(n6)};
\drawvertex{(n7)};
\drawvertex{(n8)};
\drawvertex{(n9)};

  \draw (n1) ++ (0,1) node {$O_1$};
  \draw (n2) ++ (0,1) node  {$O_2$};
  \draw (n3) ++ (0,1) node  {$s_1O_1O_2$};
  
  \draw (n4) ++ (-0.75,1) node {$O_3$};
  \draw (n5) ++ (-1,1) node {$O_4$};
  \draw (n6) ++ (2,1) node  {$s_2O_3O_4$};
  
  \draw (n7) ++ (0,-1) node  {$s_3 O_1O_3$};
  \draw (n8) ++ (0,-1) node  {$s_4 O_2 O_4$};
  \draw (n9) ++ (1,-1) node  {$s_5O_1 O_2 O_3O_4$};
\end{scope}
\end{tikzpicture}
\caption{C3 Configuration}
\label{fig:k5problems}
\end{center}
\end{figure}
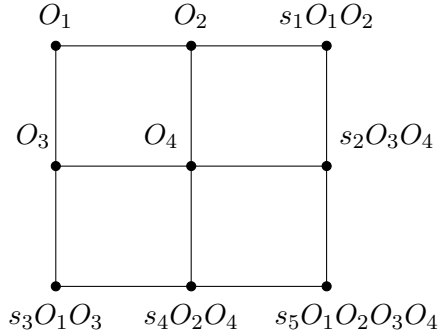
Assume that the assignments $A_p$ were done as
shown in Figure \ref{fig:k5problems}, and let $\B$ be the
set of six blocks in that figure.
Using Knuth-Bendix algorithm in Magma, the canonical form
of $\prod_{B\in\B} \prod_{p\in B} A_p$ is $(O_1O_4)^2$. 
Since the group element $(O_1O_4)^2$ is not equal to $I$ in $G_{\I}$,
we can not rule out the existence of parity proofs.

Since parity proofs may exist, we want to search for them.
Now let the labels in Figure \ref{fig:k5problems}
mean observables in $\C^{d\times d}$ for some fixed $d$.
We see that $(O_1O_4)^2\neq I$ is a necessary condition
for any assignment that produces a parity proof.
In fact the assignment $O_1=X\otimes I$, $O_2= I\otimes X$, 
$O_3=I\otimes Y$, $O_4=Y\otimes I$, where
$X,Y,Z,I$ are as in (\ref{eq-XYZI}) above,
was used by Mermin in \cite{Mer}
to give the well known proof of the Kochen-Specker theorem.
Whence  parity proofs exist for the C3 configuration.
\end{example}

\begin{example}
\label{ex-C2-config}
The graph C2 leads to the 
incidence structure that is depicted in Figure \ref{fig-C2}.
Assume that the assignments $A_p$ were done as
shown in Figure \ref{fig-C2}, and let $\B$ be the
set of six blocks in that figure.
Knuth-Bendix algorithm finds that the canonical form
of $\prod_{b\in\B} \prod_{p\in b} A_p$ is $I$, hence the left-hand side\
of  (\ref{eq-pro-pro}) equals $I$ for any admissible
assignment of observables to points of the C2 configuration.
By Lemma \ref{lem-IncS-Merm}, the C2 configuration
can not produce a parity proof.

\begin{figure}[htb]

\begin{center}
  \def\vertexscale{3}
  \def\arrowscale{1.2}
  \def\arrowstyle{latex'}
\begin{tikzpicture}[scale=0.4]

  \path (10,18) coordinate (n1);
  \path (5,18) coordinate (n2) ;
  \path (0,18) coordinate (n3);
  
  \path (5,14) coordinate (n4);
  \path (10,14) coordinate (n5) ;
  \path (14,14)coordinate (n6);
  
  \path (7.5,6) coordinate (n7) ;
  \path (5,10) coordinate (n8) ;
  \path (10,9) coordinate (n9) ;
\begin{scope}
  
 \foreach \from/\to in {n1/n2,n2/n3,n3/n8, n8/n7, n2/n4, n4/n8, n4/n5, n5/n6, n1/n5, n5/n9, n6/n9, n7/n9}
  \draw (\from) -- (\to);

\drawvertex{(n1)};
\drawvertex{(n2)};
\drawvertex{(n3)};
\drawvertex{(n4)};
\drawvertex{(n5)};
\drawvertex{(n6)};
\drawvertex{(n7)};
\drawvertex{(n8)};
\drawvertex{(n9)};

  \draw (n1) ++ (0,1) node {$O_1$};
  \draw (n2) ++ (0,1) node  {$O_2$};
  \draw (n3) ++ (0,1) node  {$s_1O_1O_2$};
  
  \draw (n4) ++ (-0.5,0.5) node {$O_4$};
  \draw (n5) ++ (-0.5,0.5) node {$O_3$};
  \draw (n6) ++ (0.5,0.5) node  {$s_2O_3O_4$};
  
  \draw (n7) ++ (0,-0.75) node  {$s_5 O_1 O_4$};
  \draw (n8) ++ (-1,-0.75) node  {$s_3 O_2 O_4$};
  \draw (n9) ++ (1,-0.75) node  {$s_4 O_1 O_3$};
\end{scope}
\end{tikzpicture}
\caption{C2 Configuration}
\label{fig-C2}
\end{center}
\end{figure}
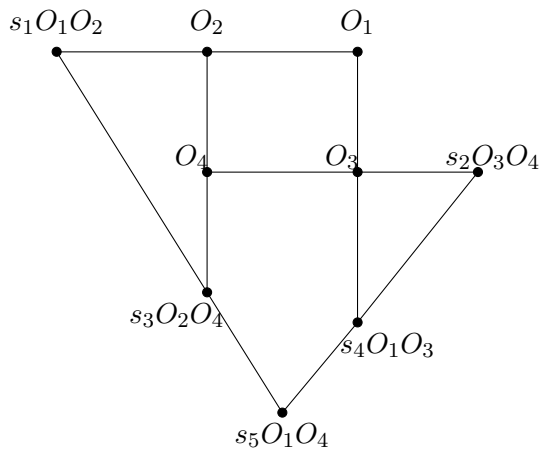
\end{example}

Using the methods outlined in this section we investigated
all incidence structures constructed from connected cubic graphs
(up to isomorphism) on up to 10 vertices.  The results are summarized
in the following statement.

\begin{proposition} 
The numbers of incidence structures
from connected cubic graphs according to producing  parity proofs
are given in the following table:

\begin{center}
\scriptsize
\begin{tabular}{|c|c|c|}
\hline
 & incidence structures  & incidence structures \\
number of vertices  &  producing   &  not producing  \\
in connected cubic graph &  parity proofs &  parity proofs  \\
\hline
\hline
4 & 0 & 1  \\
6 & 1 & 1  \\
8 & 2 & 3 \\
10 & 10 & 9 \\
\hline
\end{tabular}
\end{center}

\end{proposition} 
\begin{proof}
For the incidence structures counted in the second column we have
found an assignment of observables to the points of the incidence structure
that produces a  parity proof.
For the incidence structures counted in the third column
we proved using the implementation of Knuth-Bendix algorithm
in Magma that they can not produce a  parity proof.
\end{proof}

\section*{Acknowledgment}

Research reported in this paper was 
supported
by the Natural Sciences and Engineering Research Council of Canada
(NSERC), Collaborative Research Group ``Mathematics of Quantum
Information'' of the Pacific Institute for the Mathematical Sciences (PIMS),
Intelligence Advanced Research Projects Activity (IARPA)
and Canadian Institute for Advanced Research (CIFAR).

\appendix

%

\section{Weight distribution for the 60--105 system}

{
\scriptsize
\begin{verbatim}
> WeightDistribution(V);
[ <0, 1>, <4, 135>, <6, 810>, <8, 12195>, <9, 160>, <10, 113892>, <11, 18240>, 
<12, 1077285>, <13, 441600>, <14, 9540450>, <15, 7997824>, <16, 80906400>, <17, 
118015200>, <18, 688524520>, <19, 1448184000>, <20, 5961320616>, <21, 
15557419520>, <22, 52002701520>, <23, 147756103680>, <24, 441117024580>, <25, 
1254610425984>, <26, 3490721135520>, <27, 9499625852160>, <28, 24887073592740>, 
<29, 63507095523840>, <30, 155912963026760>, <31, 369822648368640>, <32, 
844216996941390>, <33, 1852875901104000>, <34, 3909633540468480>, <35, 
7917739173148416>, <36, 15397200649882050>, <37, 28734130298150400>, <38, 
51467429865611820>, <39, 88506321096591360>, <40, 146135139624541674>, <41, 
231792714654302400>, <42, 353282882649352920>, <43, 517597039127587200>, <44, 
729263310135826470>, <45, 988340133342723072>, <46, 1288880337830696700>, <47, 
1617684355058453760>, <48, 1954471451418300220>, <49, 2273535202515416640>, <50,
2546437247980289616>, <51, 2746415207269776000>, <52, 2852411008940091540>, <53,
2852701144397253120>, <54, 2747311965539513880>, <55, 2547589610965831680>, <56,
2274564123322337820>, <57, 1955193785568922240>, <58, 1617851718574207440>, <59,
1288608587407530240>, <60, 987792741688578932>, <61, 728611838041505280>, <62, 
517088519080163880>, <63, 352965614397949440>, <64, 231697797145211865>, <65, 
146214633571559808>, <66, 88658838120722880>, <67, 51642900930835200>, <68, 
28871970516908175>, <69, 15484467282700800>, <70, 7960297421809338>, <71, 
3916267265034240>, <72, 1843608398637195>, <73, 827932478585760>, <74, 
354477153134820>, <75, 144445514705216>, <76, 55639662848925>, <77, 
20412542826240>, <78, 6977966689330>, <79, 2267783587200>, <80, 689017459452>, 
<81, 187607370720>, <82, 55431880200>, <83, 10352153280>, <84, 4111118060>, <85,
293784576>, <86, 291511560>, <87, 1812480>, <88, 15413640>, <90, 423920> ]

Total time: 27706.459 seconds, Total memory usage: 11.03MB
\end{verbatim}
}

\end{document}